\documentclass[12pt]{article}
\usepackage{amsmath,amssymb,amsfonts}
\usepackage{amsthm}
\usepackage{mathtools}
\usepackage{fullpage}
\usepackage{cite}
\usepackage{url}
\usepackage[colorlinks=true,citecolor=blue,bookmarksnumbered=true,hyperfootnotes=true]{hyperref}

\usepackage{cleveref}
\usepackage{xspace}
\usepackage{enumitem}

\setlist{leftmargin=*} 
\setenumerate[1]{label=\normalfont(\alph*)}

\newtheorem{theorem}{Theorem}
\newtheorem{proposition}[theorem]{Proposition}
\newtheorem{corollary}[theorem]{Corollary}
\newtheorem{lemma}[theorem]{Lemma}

\crefname{theorem}{Theorem}{Theorems}
\crefname{lemma}{Lemma}{Lemmas}
\crefname{proposition}{Proposition}{Propositions}
\crefname{corollary}{Corollary}{Corollaries}

\crefname{section}{Section}{Sections}
\crefname{appendix}{Appendix}{Appendices}

\newcommand{\etal}{\textit{et al.\/}\xspace}

\newcommand{\half}{\tfrac{1}{2}}
\newcommand{\eps}{\varepsilon}

\DeclareMathOperator*{\E}{\mathbb{E}}

\newcommand{\cY}{\mathcal{Y}}

\newcommand{\Ent}{\mathcal{H}}
\newcommand{\Bh}{\mathcal{B}}

\DeclareMathOperator{\js}{\mathsf{JS}}
\DeclareMathOperator{\hl}{\mathsf{H}}

\title{A note on some inequalities used in channel polarization and polar coding}

\author{T.S.~Jayram\thanks{T.S. Jayram is with IBM Almaden Research Center, San Jose CA 95120, USA. Email: jayram@us.ibm.com}%
\and 
Erdal~Ar{\i}kan\thanks{E. Ar{\i}kan is with the Department of Electrical-Electronics Engineering,  
Bilkent University, 06800, Ankara, Turkey. Email: arikan@ee.bilkent.edu.tr}}

\date{}

\begin{document}
\maketitle

\begin{abstract}
We give a unified treatment of some inequalities that are used in the
proofs of channel polarization theorems involving a binary-input discrete memoryless channel. 
\end{abstract}


Let $W$ be a binary-input discrete memoryless channel with $W(y|x)$ denoting the transition probability that output letter 
$y \in \cY$ is received given that input $x\in \{0,1\}$ is sent.
Assume without loss of generality that the channel is non-degenerate, i.e., 
$W(y|0)+W(y|1)> 0$ for every $y \in \cY$.
Let the symmetric capacity be defined as:\footnote{$\log$ denotes the binary logarithm and $\ln$ denotes the natural logarithm.}
\[ I(W) \coloneqq \sum_{y} \sum_{x \in \{0,1\}} \half W(y|x)\log \frac{W(y|x)}{\half W(y|0) + \half W(y|1)} \]
and the Bhattacharyya parameter as:
\[ Z(W) \coloneqq \sum_y \sqrt{W(y|0)W(y|1)} \]
Below, we prove various inequalities relating the Bhattacharyya parameter to the symmetric capacity.

Let $\Ent(q) \coloneqq -q\log(q) - (1-q)\log(1-q)$ denote the binary entropy 
function. Also define the Bhattacharyya
function $\Bh(q) \coloneqq 2\sqrt{q(1-q)}$. 
Both $\Ent(q)$ and $\Bh(q)$ are concave functions whose common domain
and range are both equal to the interval $[0,1]$.
Define:
\[ \phi: u \in [0,1] 
\mapsto \Ent\Bigl(\tfrac{1-\Bh(\frac{1-u}{2})}{2}\Bigr)
= \Ent\bigl(\tfrac{1-\sqrt{1-u^2}}{2}\bigr) \]
It can be verified that $\phi$ is a bijection and that 
$\phi(\Bh(q)) = \Ent(q)$ for all $q \in [0,1]$.
Anantharam~\etal~\cite{AnantharamGKN13} studied $\phi$ in a different setting and showed that it is convex. We reprove this below and demonstrate other properties of $\phi$ that yield useful relationships between $I(W)$ and $Z(W)$ in a unified manner.

\begin{lemma} \label{lem:phi-deriv}
$0 < \phi''(u) < \phi'(u)/u$, for all $u \in (0,1)$.
\end{lemma}
\begin{proof}
Let $v = \sqrt{1-u^2} \in (0,1)$ to 
simplify the calculations. 
Taking derivatives of $\phi$ we obtain:
\begin{align}
\quad  \frac{1}{u} \cdot \frac{d\phi}{du} 
&= \frac{1}{\ln 2} \cdot \frac{\alpha(v)}{v} \label{eq:phi-derivA} \\
\frac{d^2 \phi}{du^2} 
&= \frac{1}{\ln 2} \cdot \frac{\alpha(v)-v}{v^3}, \label{eq:phi-derivB} 
\end{align}
where $\alpha(v)$ above denotes the inverse hyperbolic tangent function, i.e., 
$\alpha: v \in (0,1) \mapsto \half\log \bigl(\tfrac{1+v}{1-v}\bigr)$.

The Taylor series of $\alpha(v)$ equals 
$\sum_{n \ge 1} \tfrac{v^{2n-1}}{2n-1}$ which converges absolutely 
for $v \in (0,1)$. Therefore:
\begin{align*}
\frac{\phi'(u)}{u} &= \frac{1}{\ln 2} \cdot \Bigl(1 + \sum_{n \ge 1} \frac{v^{2n}}{2n+1}\Bigr) \\
\phi''(u) &= \frac{1}{\ln 2} \cdot \Bigl(\frac{1}{3} + \sum_{n \ge 1} \frac{v^{2n}}{2n+3} \Bigr)
\end{align*}
Comparing the right hand side of both expressions term by term, 
the desired inequality follows for all $u \in (0,1)$.
\end{proof}

\begin{lemma}\label{lem:phi-con}
The function $\phi(u)$ is strictly convex whereas the function $\phi(\sqrt{w})$ is strictly concave over their domain $[0,1]$.
\end{lemma}
\begin{proof}
Since $\phi(u)$ is continuous over its domain $[0,1]$, and $\phi''(u) > 0$
for all $u \in (0,1)$ by~\cref{lem:phi-deriv}, thus
$\phi(u)$ is strictly convex.

Define $\psi(w) := \phi(\sqrt{w})$ and let $u=\sqrt{w}$. Now
$\psi''(w) = \tfrac{1}{4u^2} \cdot \bigl(\phi''(u) - \phi'(u)/u\bigr)
< 0$ by ~\cref{lem:phi-deriv}, for all $u \in (0,1)$. 
Since $\psi(w)$ is also continuous over $[0,1]$, it is strictly concave.
\end{proof}
As a consequence, we obtain the following inequalities.
\begin{lemma} \label{lem:phi-ineq}
For all $u \in [0,1]$:
\begin{enumerate}[label=\normalfont(\alph*)]
\item $\phi(u) \le u$ with equality only at $u \in \{0,1\}$;
\item $\phi(u) \ge u^2$ with equality only at $u \in \{0,1\}$; and
\item $\phi(u) \ge 1 + (u - 1)/\ln 2$ with equality only at $u=1$.
\end{enumerate}
\end{lemma}
\Cref{lem:phi-ineq}(a) can be restated as $\Ent(q) \le \Bh(q)$, 
as shown by Lin~\cite[Theorem 8]{lin1991divergence}.
\Cref{lem:phi-ineq}(b) can be restated as $\Ent(q) \ge \Bh(q)^2$,
as shown by Ar{\i}kan~\cite{arikan2010source}.
The lower bounds given in \Cref{lem:phi-ineq}(b) and \Cref{lem:phi-ineq}(c) are incomparable: when $u=0$,~\cref{lem:phi-ineq}(b) is tight but 
not~\cref{lem:phi-ineq}(c); when $u=1-\eps$ for some small $\eps>0$,
then $\phi(u) = 1 - \eps \log e + \Theta(\eps^2)$.
Up to the linear term this matches the bound given by~\cref{lem:phi-ineq}(c) but we get a worse 
bound with~\cref{lem:phi-ineq}(b).
\begin{proof}[Proof \normalfont(of~\cref{lem:phi-ineq})]
The proof uses the convexity statements in~\cref{lem:phi-con}.
The inequality in part~(a) follows by convexity:
$\phi(u) \le (1-u) \cdot \phi(0) + u \cdot \phi(1) = u$.
Note that $\phi(u) - u = 0$ for $u\in \{0,1\}$ and
by strict convexity of the function $\phi(u)-u$, this value 
is achieved only at the end points.

The inequality in part~(b) follows by concavity:
$\phi(\sqrt{w}) \ge (1-w) \cdot \phi(\sqrt{0}) + w \cdot \phi(\sqrt{1}) = w$; now set $w = u^2$.  
By strict concavity, the minimum of $\phi(\sqrt{w})-w$
is achieved only at the end points so equality holds 
only at $w = u \in \{0,1\}$.    

For part~(c), let $\ell(u)$ denote the right side of the inequality.
We show that $\ell(u)$ is the tangent line at $u=1$ which by convexity would establish the inequality.
By definition the tangent at $u=1$ equals $\phi(1) + (u-1)\phi'(1)$ so
we need to show that $\phi'(1) = \tfrac{1}{\ln 2}$.
By~\cref{eq:phi-derivA}, we have:
\begin{align*}
\phi'(1) 
&= \lim_{u \to 1} \frac{\phi'(u)}{u} 
= \lim_{x \to 0} \frac{\alpha(x)}{x \ln 2} \\
&=  \frac{1}{\ln 2} \cdot  \lim_{x \to 0} \alpha'(x) 
= \frac{1}{\ln 2} \cdot \lim_{x \to 0}\frac{1}{1-x^2}
= \frac{1}{\ln 2}
\end{align*}
Now $\phi(u) = \ell(u)$ at $u=1$ and by strict convexity
of $\phi(u)-\ell(u)$, 
its minimum is achieved only at this point.
\end{proof}

The above properties of $\phi$ have the following implications for
relating $I(W)$ to $Z(W)$.
Under the uniform distribution on the input $\{0,1\}$, let $Y$ denote
the output induced by the channel, i.e.,  
for each output letter $y \in \cY$, $p_Y(y) = \half(W(y|0)+W(y|1))$. 
Define the random variable:  
\[ U(y) \coloneqq \Bh(Q(y)),
\quad \text{where} \ 
Q(y) \coloneqq \frac{W(y|0)}{W(y|0)+W(y|1)} \]
The law of $Q$ is referred to as the Blackwell measure of $W$ in \cite{Raginsky16}. 
Related measures, giving alternative characterizations of a binary-input memoryless channel, have been used extensively in the context of information combining in \cite[Ch.~4]{mct},
and more specifically in polar coding in \cite[p.~30]{sasoglubook}.

Rewrite the channel parameters $I(W)$ and $Z(W)$ as expectations 
of appropriate functions of $U$:
\begin{equation}
\begin{split}
Z(W) &= \sum_y p_Y(y) \Bh(Q(y)) = \E \Bh(Q) = \E U  \\
1 - I(W) &= \sum_{y} p_Y(y) \Ent(Q(y)) = \E \Ent(Q) = \E \phi(U)
\end{split} \label{eq:chan-rewrite}
\end{equation}

\begin{theorem}\label{thm:channel-bounds}
$Z(W) \ge 1 - I(W) \ge \phi(Z(W))$
\end{theorem}
\begin{proof}
Applying~\cref{lem:phi-ineq}(a) and then using the fact that 
$\phi$ is convex (\cref{lem:phi-con}) yields: 
$\E U \ge \E \phi(U) \ge \phi( \E U)$.
Now substitute the identities in~\cref{eq:chan-rewrite}.
\end{proof}

By~\cref{lem:phi-ineq}, the first inequality is tight iff
$U \in \{0,1\}$ with probability 1. In other words, the inequality is tight iff the channel $W$ is such that
$W(y|0)W(y|1) = 0$ or $W(y|0) = W(y|1)$ for each output $y$. A channel with this property is called a binary erasure channel (BEC).
Indeed, this inequality was proved by Ar{\i}kan~\cite[Prop.~11]{arikan2009channel} by 
an indirect argument, using an extremal property of the BEC
in channel polarization.

The second inequality is tight iff $U$ is constant
with probability 1. Divide the outputs into two classes based
on the predicate $W(y|0) > W(y|1)$; this is operationally 
equivalent to a binary symmetric channel (BSC), 
i.e., a binary-input channel for which there exists 
a constant $0 \le \epsilon \le \half$ such that each $y$ satisfies
$\epsilon \cdot W(y|x) = (1-\epsilon) \cdot W(y|1-x)$ for some 
$x \in \{0,1\}$.

Now~\cref{lem:phi-ineq}(b) implies that $\phi(Z(W) \ge Z(W)^2$
so we obtain: $1 - I(W) \ge Z(W)^2$ (cf.~\cite{arikan2010source}). Equality
holds only when $Z(W) \in \{0,1\}$. Equivalently, the distributions
$W(\cdot|0)$ and $W(\cdot|1)$ are either identical or have disjoint
support. 
Next~\cref{lem:phi-ineq}(c) implies that 
$I(W) + Z(W) \cdot \log e \le \log e$. Equality holds only
when $Z(W) = 1$, i.e., the distributions
$W(\cdot|0)$ and $W(\cdot|1)$ are identical. 
To summarize:
\begin{corollary}\label{cor:channel-ineq}
For a binary input symmetric channel $W$:
\begin{enumerate}
\item $I(W) + Z(W) \ge 1$. Equality holds  only for the BEC.
\item $I(W) + \phi(Z(W)) \le 1$. Equality holds  only for the BSC.
\item $I(W) + Z(W)^2 \le 1$. Equality holds iff $Z(W) \in \{0,1\}$. 
\item $I(W) \cdot \ln 2 + Z(W) \le 1$. Equality holds iff $Z(W) = 1$.
\end{enumerate}
\end{corollary}
Finally, we note that these inequalities can be restated in terms of 
distances between probability distributions, which was the original
motivation of Lin~\cite{lin1991divergence}.
Let $P$ and $Q$ be two distributions $P$ and $Q$ on $\cY$.
Identify $W(\cdot|0)$ with $P$ and $W(\cdot|1)$ with $Q$.
Then the Hellinger distance $\hl(P,Q)$ equals
$\sqrt{1-Z(W)}$ and the Jensen--Shannon divergence $\js(P,Q)$
equals $I(W)$. Thus~\cref{cor:channel-ineq} can be
restated as follows:
\begin{proposition}
For two distributions $P$ and $Q$:
\[ \hl^2(P,Q) \le \js(P,Q) \le \hl^2(P,Q) \cdot 
       \min\bigl(\log e, 2-\hl^2(P,Q)\bigr) \]
\end{proposition}

\section*{Acknowledgment}
This work was jointy done at the Simons Institute for the
Theory of Computing at UC Berkeley. The authors would like to thank the institute for their invitation to participate in the Information Theory Program during Jan. 2015 --  June 2016.

\end{document}